\documentclass[11pt]{article}

\usepackage[]{amsmath,amssymb,amsfonts,latexsym,amsthm,enumerate,url}
\date{}

\usepackage{amsmath, amsthm, amssymb}
\usepackage{geometry}
\usepackage[dvips]{graphics}
\newcommand{\bra}[1]{\langle #1|}
\newcommand{\ket}[1]{|#1\rangle}
\newcommand{\braket}[2]{\langle #1|#2\rangle}

\newcommand{\Tr}{\textrm{tr}}
\newcommand{\vmin}{\textrm{vmin}}
\newcommand{\vmax}{\textrm{vmax}}
\newcommand{\saux}{\textrm{aux}}

\newcommand{\sizeD}{\textrm{D}}

\newcommand{\ups}{\textrm{up}}
\newcommand{\downs}{\textrm{down}}

\newcommand{\failprob}{\textrm{FailProb}}

\newcommand{\fail}{\textrm{fail}}

\newtheorem{thm}{Theorem}[section]

\newtheorem{lem}[thm]{Lemma}
\theoremstyle{remark}

\theoremstyle{definition}

\begin{document}
  \title{Quantum Multi Prover Interactive Proofs with Communicating Provers \\ {\small Extended Abstract}}
  \author{Michael Ben Or \thanks{benor@cs.huji.ac.il, The Hebrew University, Jerusalem, Israel}
  \quad{Avinatan Hassidim \thanks{avinatan@cs.huji.ac.il, The Hebrew University, Jerusalem, Israel.
  Part of the work was done while visiting the Perimeter Institute and MIT}}
  \quad{Haran Pilpel \thanks{haranp@math.huji.ac.il, The Hebrew University, Jerusalem, Israel}}}
  \maketitle

\textheight 7.62in

\maketitle

\thispagestyle{empty}

\begin{abstract} Multi Prover Interactive Proof systems (MIPs)
were first presented in a cryptographic context, but ever since
they were used in various fields. Understanding the power of MIPs
in the quantum context raises many open problems, as there are
several interesting models to consider. For example, one can study
the question when the provers share entanglement or not, and the
communication between the verifier and the provers is quantum or
classical. While there are several partial results on the subject,
so far no one presented an efficient scheme for recognizing NEXP
(or NP with logarithmic communication), except for \cite{KM03}, in
the case there is no entanglement (and of course no communication
between the provers).

We introduce another variant of Quantum MIP,
where the provers do not share entanglement, the communication
between the verifier and the provers is quantum, but the provers
are unlimited in the {\em classical} communication between them.
At first, this model may seem very weak, as provers who exchange
information seem to be equivalent in power to a simple prover.
This in fact is not the case---we show that any language in NEXP
can be recognized in this model efficiently, with just two provers
and two rounds of communication, with a constant
completeness-soundness gap.

The main idea is not to bound the information the provers exchange
with each other, as in the classical case, but rather to prove
that any ``cheating'' strategy employed by the provers has
constant probability to diminish the entanglement between the
verifier and the provers by a constant amount. Detecting such
reduction gives us the soundness proof. Similar ideas and
techniques may help help with other models of Quantum MIP,
including the dual question, of non communicating provers with
unlimited entanglement.
\end{abstract}

\newpage
\pagebreak \setcounter{page}{1}


\section{Introduction}\label{sec:intro}
%
Multi Prover Interactive Proofs (MIPs) have been studied
extensively in the classical setting, and provide an exact
characterization of NEXP \cite{BFL92}. Extending MIPs to the
quantum setting poses many important open problems, and may give
us more intuition regarding the power of entanglement. There are
several possible generalizations for quantum multi-prover schemes,
which differ in the power of the verifier (which can be quantum or
classical), and in the relation between the provers (for example,
how much entanglement they have).  The first results for this
problem were given by Kobayashi and Matsumoto \cite{KM03}. They
proved that as long as the provers share a bounded (polynomial)
amount of entanglement, the set of languages which can be
recognized is contained in NEXP (even if the verifier is quantum).

We do not understand the power of the model when the verifier is
classical and the provers share (limited or unlimited)
entanglement. In particular, Cleve et al. \cite{CHTW04} provide
examples where the proof is valid if the provers share no
entanglement, but is no longer sound when they do. Preda
\cite{Pre} showed that if the provers are not limited to quantum
entanglement, but instead have an unlimited amount of nonlocal
boxes \cite{PR97}, then the set of recognizable languages is
contained in EXP.

There are also some positive results when the provers are quantum.
Cleve et al. \cite{CGJ07} provide a proof system for NP when the
verifier is classical and the provers who share an unlimited
amount of entanglement. The proof scheme provides a constant gap,
but the communication is linear. Kempe et al. \cite{KKMTV07} give
a quantum protocol for recognizing languages in NP by a quantum
verifier with logarithmic communication, when the provers share
unlimited entanglement. However, when $x \notin L$ the probability
that the verifier will discover this is $1 - O(1/n)$, which means
that it is necessary to repeat the protocol a polynomial number of
times to get constant soundness. Ito et al. \cite{IKPSY07} use
this result, and give a $3$ prover proof system for NEXP which is
resistant to entanglement with soundness of just $1 - 2^{-poly}$.

\subsection{Our Results}
An important assumption underlying the work on multi prover
schemes is that the provers are not allowed to pass information
between themselves. The results of \cite{KW00,Pre} could lead us
to believe that a proof system with a quantum verifier and two
provers who can pass classical information between them is limited
to EXP. Surprisingly, this is not the case (assuming EXP $\neq$
NEXP). We show that:

\begin{thm} \label{main-thm}
Let $V$ be a polynomial time verifier that can exchange quantum
messages with two computationally unbounded provers. The provers
share no entanglement, but can freely communicate classically
between them. Then for any $L \in $ NEXP there is a two round
protocol for the verifier and provers such that for any string $x$
\begin{itemize}
\item (completeness) If $x \in L$ then there are two prover
strategies such that $V$ will accept $x$ with probability 1.

\item (soundness) If $x \notin L$ then for any two prover
strategies the probability that $V$ will accept $x$ is at most $c$
for some constant $0 < c < 1$.
\end{itemize}

The communication between the verifier and the provers is is
polynomial in the length of the input\footnote{Equivalently we can
state our result for NP, bounding the communication to be
logarithmic.}.
\end{thm}

We note that augmenting the provers in our model with unlimited
entanglement gives something which is contained in EXP \cite{KW00}
(as this is equivalent to quantum communication and thus to a
single quantum prover). Bounding the verifier to be classical,
would limit us to languages in PSPACE \cite{Sha90} (as in this
scenario is equivalent to a single prover and a classical
verifier), so both conditions are necessary. This problem is in a
way dual to the scenario where the provers do not have any means
of communication but instead have unlimited entanglement, where
much less is known.


Quantum MIPs are thought to be a model of computation which may
give us better understanding of entanglement, and its powers.
Surprisingly, our result, which is stated in a model with no
entanglement between the provers, is based on following the
entanglement between the provers and the verifier. Each message
the verifier sends is a superposition of two classical queries.
Measuring the message would ruin the superposition, and will be
caught by the verifier. However, a strategy which does not measure
it ``enough'' does not extract enough useful classical
information, and prevents the provers from coordinating answers
via the classical channel. Most of the paper follows the amount of
entanglement between the verifier and the provers during the
protocol, making sure that either the provers do not extract
enough information to answer with very high probability (we note
that from an information-theoretic point of view they extract many
bits of information--so we use tailored bounds), or they have
some chance of getting caught.


%
%

\subsection{Related Work}
It is interesting to view the results of this paper in light of
the complexity class QMA(2), defined by Kobayashi, Mastumoto and
Yamakami \cite{KMY01}. Intuitively, this is the class of languages
which can be recognized by a quantum verifier with two unentangled
bounded pieces of quantum evidence. While there is no classical
analog for this problem (having two classical witnesses is still
NP), there is evidence that QMA(2) strictly contains QMA
\cite{LCV07}. Blier and Tapp \cite{BT07} showed that a verifier
can recognize an NP complete language with soundness $1 -
O(1/n^6)$. A constant soundness completeness gap in their results
would imply our own. We note however, that Aaronson et al.
\cite{ABDFS08} give evidence towards QMA(2)$\subseteq$PSPACE, and
therefore we do not expect that this is the case.

The idea of using Private Information Retrieval
\cite{CGKS95,KdW03,KdW04} schemes (PIRs) has been suggested by
Cleve et al. \cite{CGJ07}. Our protocol is in a sense a cheat
sensitive PIR where the verifier can check whether the prover has
tried to learn information. A similar quantum PIR scheme has been
independently presented by \cite{GLM07} in a different context. It
is important to note that information disturbance tradeoffs
proposed by such quantum PIR schemes are by themselves insufficient to
prove the soundness of our multi-prover protocol, since the
leakage of even a small amount of information might enable the
provers to succeed in cheating the verifier.

\section{Preliminaries}
We assume the reader is familiar with quantum computation (see
\cite{NC00} for example).

Let $L \in $ NEXP. By standard PCP machinery, we can assume that
given $x$ the verifier has an implicit efficient access to an
exponentially long 3-SAT formula $\Phi$, such that if $x \in L$
then $\Phi$ is satisfiable, and otherwise any assignment can satisfy at
most a $1 - \gamma$ proportion of the clauses for some constant
$\gamma > 0$. We can also assume that each variable appears
exactly 5 times, and each clause contains three different
variables. Let $C$ denote the set of clauses and $V$ the set of
variables. If a variable $v \in V$ appears in a clause $c\in C$ we
write $v \in c$. Let $M = |C|$ denote the number of clauses and $N
= |V|$ the number of variables. Let $T$ be a truth assignment for
$\Phi$. For a variable $x \in V$, let $T(x)$ denote the value $T$
assigns $x$. For a clause $y \in C$, if $y$ contains the variables
$v^y_1, v^y_2, v^y_3$, let $T(y) = T(v^y_1), T(v^y_2), T(v^y_3)$.


Alice (Bob) has a private Hilbert space $H^{p}_{A}$ ($H^{p}_{B}$),
with some finite arbitrarily large dimension $d$ (we assume
without loss of generality that the dimensions are identical). The
messages between Alice (Bob) and the verifier will be sent by
passing a state which is in a Hilbert space $H^{m}_{A}$
($H^{m}_{B}$). For convenience, we partition the private Hilbert
space of the verifier into three parts, $H_v = H^{\saux}_v \otimes
H^{v}_A \otimes H^{v}_B$. The Hilbert spaces $H^{v}_A, H^{v}_B$
will be used with messages sent to different provers, but they are
private spaces that belong to the verifier. We let the verifier
send and receive classical messages from Alice\footnote{This can
be done by using a larger space $H^{m}_{A}$, with the verifier
measuring the part of the space which should be used for the
classical message. Thus, this does not change the model, and is
only done for clarity.}. For the protocol we present, the
dimensions of the Hilbert spaces used are $\dim({H^{m}_{A}}) =
8M$, which would fit a clause $y$ and $T(y)$, $\dim(H^{v}_A) = M$,
$\dim({H^{m}_{B}}) = 2N$ which would fit a variable and the value
it is assigned, and $\dim(H^{v}_B) = N$.

\section{Algorithm}

Let $\pi$ be a probability distribution which chooses two clauses
$y, \tilde y$ uniformly at random from $C$, and two variables $x,
\tilde x$ uniformly at random from $V$, with the constraint that $x$
appears in $y$ ($\tilde x$ does not necessarily appear in $\tilde
y$).



{\bf Protocol for verifier}

\begin{enumerate}

\item Sample $\pi$ to get $y, \tilde y, x,  \tilde x$. Generate
the states on $O(\log(N))$ qubits
\[
\frac{1}{\sqrt{2}} (\ket{yy} + \ket{\tilde y \tilde y}) \otimes
\ket{000} \in H^{v}_A \otimes H^{m}_{A}
\]
\[
\frac{1}{\sqrt{2}} (\ket{xx} + \ket{\tilde x \tilde x}) \otimes
\ket{0} \in H^{v}_B \otimes H^{m}_{B}
\]
Send Alice (Bob) the message space $H^{m}_{A}$ ($H^{m}_{B}$),
which consists of the last $m+3$ ($n+1$) qubits.

\item Let $T$ be a satisfying assignment for $\Phi$ (if one
exists). Alice should apply the unitary which takes $\ket{c}
\otimes \ket{000} \rightarrow \ket{c} \otimes \ket{T(c)}$ for any
clause $c \in C$, and Bob should apply the unitary which takes
$\ket{v} \otimes \ket{0} \rightarrow \ket{v} \ket{T(v)}$ for $v
\in V$. In fact the provers apply any Local Operations and
Classical Communication protocol they want among themselves.
Finally, Alice (Bob) returns the verifier the message space
$H^{m}_{A}$ ($H^{m}_{B}$).

\item Send Alice the classical values $y, \tilde y, x, \tilde x$.
Alice returns $8$ bits: $T(y), T(\tilde y), T(x), T(\tilde x)$. If
Alice returned quantum values, the verifier measures them according
to the standard basis.

\item Verify that the clause $y$ is satisfied, and that $T(x)$
matches $T(y)$. Perform the SWAP test \cite{BCWW01} between the
state in $H^{v}_A \otimes H^{m}_A$ and $\frac{1}{\sqrt{2}}
(\ket{yy} \otimes \ket{T(y)} + \ket{\tilde y \tilde y} \otimes
\ket{T( \tilde y)}$ and between the state in $H^{v}_B \otimes
H^{m}_B$ and $\frac{1}{\sqrt{2}} (\ket{xx} \otimes \ket{T(x)} +
\ket{\tilde x \tilde x} \otimes \ket{T( \tilde x)}$. Accept if all
tests passed.

\end{enumerate}
Note that the verifier does not generate any entanglement between
the provers. This means that it is possible to repeat the protocol
in order to reduce the error probability.

Completeness: With a common satisfying assignment the provers can
apply the required quantum transformation, and all the tests will
be passed with probability 1.

\section{Soundness of the Protocol}\label{sec-soundness}
\noindent{\bf Intuition} To simplify the analysis, we modify the
protocol. First, we purify the verifier. This will enable us to
talk about the probability of a set of queries given measurements
by the provers. The second modification will be to strengthen the
provers, allowing them to perform any joint separable measurement
instead of Local Operations and Classical Communication (LOCC),
which will enable us to write the state after their actions. We
prove that the provers have a constant failure probability for any
result $k$ of the separable measurement they make. We begin by
finding an estimate for the probability that the verifier measures
$(y, \tilde y, x , \tilde x)$ as a function of the provers' result
$k$. Next, we show that if $k$ is more probable given a clause
$y_1$ then given another result $y_2$, and the verifier measured
$(y_1, y_2, x , \tilde x)$ for any $x, \tilde x$, then there is
constant probability that Alice fails the SWAP test (because such
a measurement operator diminishes the entanglement between $H^v_A$
and $H^m_A$).

We then show that either the measurement has a constant
probability to diminish the entanglement, or after it there is
still a large set of clauses (and variables which appear in them)
which are all ``almost uniformly'' probable. The set will be large
enough that no assignment will satisfy all of it. This means that
if the provers succeed with very high (but constant) probability,
they must succeed on a large portion of this ``uniform'' set, and
thus they must succeed on a very large number of clauses and
variables.  This will give a strategy for the classical protocol
which has success probability greater then $1 - \gamma/3$, which
is a contradiction.

{\bf The Modified Protocol} As stated above, the first
modification is to purify the sampling of $\pi$, postponing it
until after the provers act on the information. It uses
$H^{\saux}_v$ with $\dim(H^{\saux}_v) = M^2 N^2$. The verifier
generates
\[
\psi_{\pi} = \sum_{y,  \tilde  y \in C} \sum_{x \in y} \sum_{ \tilde
x \in V} \ket {y \tilde y, x  \tilde  x} \otimes \frac{1}{\sqrt{2}}
(\ket{yy} + \ket{ \tilde y \tilde y}) \otimes \ket{000} \otimes
\frac{1}{\sqrt{2}} (\ket{xx} + \ket{ \tilde x \tilde x}) \otimes
\ket{0} \in H^{\saux}_v \otimes H^{A}_v \otimes H^A_m \otimes H^{B}_v
\otimes H^B_m
\]
As before, the verifier sends Alice (Bob) the Hilbert space
$H^{A}_m$ ($H^{B}_m$). After Alice and Bob act on the message
spaces they get and return $H^{A}_m, H^{B}_m$, the verifier
measures $H^{\saux}_v$ to get $y, \tilde y, x , \tilde x$ and sends
them to Alice as in Protocol 1. This modification does not change
the cheating power of the provers (they cannot tell what protocol
is being used).

The second modification is to replace the LOCC done by the provers
in the first stage with a single joint separable measurement.
\cite{BDFMRSSW98,BNS97} proved that this is strictly stronger than
LOCC. In particular they showed how to transform any LOCC protocol
into such a measurement. As the provers are not entangled, we can
assume that their private spaces are initialized with the state
$\ket{0 \ldots 0}$. Letting $\rho = \ket{\psi_{\pi}}
\bra{\psi_{\pi}}$, the provers' operation now becomes applying a
measurement with operators

\[ (I_{M^2N^2} \otimes I_{M} \otimes A_k \otimes I_N \otimes
B_k)^\dagger (I_{M^2N^2} \otimes I_{M} \otimes A_k \otimes I_N
\otimes B_k) \]

\noindent where $I_p$ is the $p \times p$ identity matrix, $A_k$
is an $8Md \times 8Md$ matrix, $B_k$ is a $2Nd \times 2Nd$ matrix
and
\[ \sum_k (A_k \otimes B_k)^\dagger (A_k \otimes B_k) = I_{16NMd^2} \]
The Hilbert spaces $H^{A}_m, H^{B}_m$ are then returned to the verifier.

We now calculate the probability that the verifier measured values
$r = (y, \tilde y, x, \tilde x)$, conditioned on the fact that the
measurement result was $k$. Denote $A_k(y) = \Tr(A_k (\ket{y}
\bra{y} \otimes I)A_k)$, where we are tracing over the private
qubits of the prover and the qubits which define the assignment,
and similarly $B_k(x) = \Tr(B_k (\ket{x} \bra{x} \otimes I)x_k)$.
In Appendix \ref{calc-prob}, we prove that for $y \neq \tilde y, x
\neq \tilde x$

\begin{equation} \label{exact-prob-of-tuple}
\Pr (y, \tilde y, x, \tilde x |k) = \frac{(A(y) + A(\tilde
y))(B(x) + B(\tilde x))}{\sum_{c, \tilde c \in C, v \in c, \tilde
v \in V} \Pr (c, \tilde c, v, \tilde v |k)}
\end{equation}

\noindent where if $y = \tilde y$ the numerator changes to
$4A(y)(B(x) + B(\tilde x))$, and similarly for $x, \tilde x$.

\noindent We give some intuition for Equation
(\ref{exact-prob-of-tuple}). The numerator is the product of two
factors, because when the verifier measures before the provers
(which is physically equivalent) the provers are unentangled, and
therefore the probability of $k$ is just the $\Tr(A_k \rho
A_k^{\dagger}) \cdot \Tr(B_k \rho B_k^{\dagger})$. Alice's factor
is composed of two terms, because tracing out the verifier Alice
just gets a mixed state of $\frac{1}{2} \ket{y}\bra{y} +
\frac{1}{2} \ket{\tilde y}\bra{\tilde y}$.

Omitting the subindex $k$, and denoting $W_{A_k} = \sum_{i}A_k(i)
= \Tr(A_k), W_{B_k} = \sum_{i}B_k(i) = \Tr(B_k), \tilde W =
\Sigma_{c \in C, v\in c} A_k(c)B_k(v)$ We show the following bound
in In Appendix \ref{calc-prob}, by bounding the denominator

\begin{equation} \label{prob-bound}
\Pr (y, \tilde y, x, \tilde x |k) \geq \frac{A(y)B(x) +
A(y)B(\tilde x) + A(\tilde y)B(x) + A(\tilde y)B(\tilde x)}{2 MN
\tilde W + 22 M W_{A} W_{B}}
\end{equation}

\subsection{Auxiliary Lemmas}
We show that if $A_k$ is too skewed, then for certain values of
$y, \tilde y$, Alice has a good chance of failing the SWAP test.
Formally:

\begin{lem} \label{alice-damage}
Assume $A(y) \ge p A(\tilde y)$, $p > 1$. Then for any assignment
$T$, the probability that the verifier will catch Alice cheating
in the SWAP test is at least $\frac{1}{2} - \frac{\sqrt{p}}{1+p}$.
\end{lem}

The proof is found in Appendix \ref{app-alice-damage}, as it is somewhat
technical. The intuition is that the super-operator which
acts on the state diminishes the entanglement between $H^v_A$ and
$H^m_A$. Therefore, this is true for any assignment Alice will
send in the second round of the protocol.

If the condition of lemma \ref{alice-damage} holds, we say that
the measurement {\it $p$-damaged} the state. An analogous lemma
holds for Bob. The following lemma is trivial:
\begin{lem}
If there exists a set $D \subset Y \times \tilde Y \times X \times
\tilde X$ such that
\begin{enumerate}
\item For any $d = (y, \tilde y, x, \tilde x) \in D$ we have $x
\in y$, and either $A$ or $B$ $p$-damage $d$ for some constant
$p$.
\item  $\sum_{d \in D} \Pr(d | k) > \epsilon_{\sizeD}$ for some
constant $\epsilon_{\sizeD}$
\end{enumerate}

Then at least one of the provers gets caught in the SWAP test with
probability $\epsilon_{\sizeD} \left(\frac{1}{2} -
\frac{\sqrt{p}}{1+p}\right)$.
\end{lem}
%

In this case we say that D is an $(\epsilon_{\sizeD}, p)$ {\it bad
set}.

\subsection{Large $NM\tilde W$}
\begin{thm} \label{large-NMW}
If $NM \tilde W \ge 100 M W_{A} W_{B}$ then at least one of the
provers fails the SWAP test with probability $\frac{1}{6.96 \cdot
10^9} = \min \left\{ \frac{1}{6.96 \cdot 10^9}, \frac{1}{4.2 \cdot
10^7} \right\}$.
\end{thm}

The proof is by contradiction. We prove Lemma \ref{steps}, which
states that if $A$ and $B$ do not have a certain property then the
provers have a constant probability of getting caught. We then
prove that if $A$ and $B$ do have that property than either a
second property holds or the provers get caught, with some
probability. The second property implies $NM \tilde W < 100 M
W_{A} W_{B}$, which is a contradiction. Remember $\tilde W =
\sum_{c \in C. v \in c} A(c)B(v)$. For $c \in C$, let $u(c) =
\Sigma_{v \in c} A(c)B(v)$, and for $S \subset C$, $U(S) = \sum_{c
\in S} u(c)$. Let
\[
S_{i} = \left\{c : \frac{\tilde W}{2^{i+1}} < u(c) \le
\frac{\tilde W}{2^{i}} \right\}
\]

\begin{lem} \label{steps}
If there exists an index $j$ such that $\sum_{i=0}^{j-1} U(S_i)
>\tilde W / 100$ and $\sum_{i=j+1}^{\infty} U(S_i) > \tilde W /
100$, then the provers get caught with constant probability
$\frac{1}{6.96 \cdot 10^9}$, generated from a $(\frac{1}{4.8 \cdot
10^{7}}, \sqrt{2})$ bad set.
\end{lem}

The proof is found in Appendix \ref{app-large-nmw}. It follows by
constructing a bad set, such that the clauses (and variables) in
$\cup_{i=0}^{j-1} S_i$ stand for $y,x$, and the clauses (and
variables) in $\cup_{i=j+1}^{\infty} S_i$ stand for $\tilde y,
\tilde x$, where we use the fact that if $u(c_1) > 2 u(c_2)$ for
some two clauses, then either Alice damages the state because
$A(c_1)> \sqrt{2} A(c_2)$, or Bob $\sqrt{2}$ damages the state, or
both of them do. We note that if $u(c_1) > 2 u(c_2), u(c_1) > 2
u(c_3)$ it may still be the case that $A(c_2) > A(c_1)$, and for
each $v_3 \in c_3$ and each $v_1 \in c_1$ $B(v_3) > B(v_1)$.
However, taking out such tuples only diminishes the size of the
bad set $D$ by a factor of 36 (a factor of 9 comes from choosing
one of the variables in the clause, and a factor of 4 comes from
choosing the clause).

If the condition of Lemma \ref{steps} does not hold, then there
must be an index $j$ such that $U(S_j) + U(S_{j+1}) > 0.98 \tilde
W$. Define $F = S_j \cup S_{j+1}$. Remembering that $W_A = \sum_{c \in
C} A(c)$, we partition the clauses in $F$:
\[
T_{i} = \{c \in F: \frac{W_{A}}{2^{i+1}} < A(c) \le
\frac{W_A}{2^{i}} \}
\]
\begin{lem} \label{T-steps}
If there exists an index $j$ such that $\sum_{i=0}^{j-1} |T_i| >
|F|/ 100$, and $\sum_{i=j+1}^{\infty} |T_i| > |F| / 100$, then the
first prover gets caught with constant probability $\frac{1}{4.2
\cdot 10^7}$, generated from a  $(\frac{1}{1.2 \cdot 10^{6}},2)$
bad set.
\end{lem}

\noindent The proof appears in Appendix \ref{app-large-nmw}. It is
very similar to the one of Lemma \ref{steps}, but much simpler.

As before, if the condition of Lemmas \ref{steps} and
\ref{T-steps} do not hold, then
\[ \exists i : |T_i| + |T_{i+1}| > 0.98 |F| \ge 0.98^2 M
> 0.96M \]
Let $G = T_i \cup T_{i+1}$. As $G \subset F$, and as $\forall c_1,
c_2 \in F : u(c_1) < 4 u(c_2)$ we have
\begin{equation}\label{lemma1}
U(G)> 0.25 \cdot 0.98 U(F) > 0.25 \cdot 0.98^2 \tilde W
\end{equation}
Note $\sum_{c \in G} \sum_{v \in c} B(v) \le 5 W_{B}$, as each
variable appears $5$ times. Also, since $\forall c \in G : A(c)
> W_A / 2^{i+1}$
\begin{equation}\label{lemma2}
0.96 M \frac{W_A}{2^{i+1}} < \frac{W_A}{2^{i+1}} |G| < \sum_{c \in
G} A(c) < W_{A}
\end{equation}
Putting this together, we get
\begin{eqnarray*}
& & 0.25 \cdot 0.98 ^2 \tilde W \stackrel{(\ref{lemma1})}{<} U(G)
= \sum_{c \in G} u(c) = \sum_{c \in G}\sum_{v \in c} A[c]B[v] \le
\sum_{c \in G}\sum_{v \in c} \frac{W_{A}}{2^{i-1}} B[v] = \\
& & \frac{4W_{A}}{2^{i+1}} \sum_{c \in G} \sum_{v \in c} B[v] \le
\frac{20 W_A W_B}{2^{i+1}} \stackrel{(\ref{lemma2})}{\le} \frac{20
W_A W_B}{0.96M}\le \frac{20 W_A W_B}{0.96N}
\end{eqnarray*}
This is a contradiction to $MN \tilde W \ge 100 M W_{A} W_{B}$.
This proves Theorem \ref{large-NMW}.

\subsection{Small $NM\tilde W$}
In this subsection we handle those values of $k$ for which the
premise of Theorem \ref{large-NMW} does not hold, i.e., $NM \tilde
W < 100 M W_{A} W_{B}$. Define $S_i = \{ c \in C :
\frac{W_A}{2^{i+1}} \le A(c) < \frac{W_{A}}{2^i} \}$. For a set $S
\subset C$, let $W(S) = \Sigma_{c \in S} A(c)$.

\begin{lem} \label{small-nmw-clause-steps}
If $NM\tilde W < 100 M W_A W_B$ and there exists an index $i$ such
that
\begin{equation}
\sum_{j = 0}^{i-1} W(S_j) > \gamma 10^{-4} W_A \bigwedge \sum_{j =
i + 1}^{\infty} |S_j| > \gamma 10^{-4} M
\end{equation}
then Alice is caught cheating with probability
$\frac{\gamma^2}{2.6 \cdot 10^{12}}$, generated from a
$(\frac{\gamma^2}{7.4 \cdot 10^{10}}, 2)$ bad set.
\end{lem}
\noindent The proof is found in Appendix \ref{app-small-nmw}. It
is very similar to that of Lemma \ref{T-steps}.

\begin{lem} \label{small-nmw-lem2}
If $NM\tilde W < 100 M W_A W_B$ and the second condition of Lemma
\ref{small-nmw-clause-steps} does not hold, then there exists an
index $i$ such that for $F = S_{i} \cup S_{i+1}$ we have
\[ |F| \ge (1 - 0.0002 \gamma) M \bigwedge W(F) \ge (1 - 0.0002 \gamma) W_A \bigwedge \forall c \in F
: A(c) \ge \frac{W_{A}}{5M}
\]
\end{lem}
\noindent Again the proof is found in Appendix
\ref{app-small-nmw}. Using $B$ instead of $A$, we define the sets
$T_i$ analogously to $S_i$: $T_i = \left\{ v \in V :
\frac{W_B}{2^{i+1}} \le B(v) < \frac{W_{B}}{2^i} \right\} $

\begin{lem} \label{bob-caught}
Either Bob gets caught cheating with probability
$\frac{\gamma^2}{3.9 \cdot 10^{12}}$ which is generated from a
$\frac{\gamma^2}{1.1 \cdot 10^{11}},2$ bad set, or else there
exists an index $i$ such that for $G = T_{i} \cup T_{i + 1}$ we have
$|G| > (1 - 0.0002 \gamma) N$,
$\Sigma_{v \in G} B(v) \ge (1 - 0.0002 \gamma) W_{B}$
and for each $v \in G$, $B(v) \ge \frac{W_{B}}{5N}$.
\end{lem}

\noindent The proof is very similar to the argument for Alice. It
is found in Appendix \ref{app-small-nmw}

Let $H = \{c \in F : \forall v \in c, v \in G\}$. As $|G| \ge (1 -
0.0002 \gamma)N$, and each variable appears 5 times we have $|H|
\ge (1 - 0.0002 \gamma)M - 5 \cdot 0.0002 \gamma N \ge (1 - 0.002
\gamma) M$. We now prove that a good success probability for Alice
and Bob implies a good success probability for the provers in the
classical game, with no communication. As the classical success
probability is bounded, this will give a bound for the quantum
success probability. Before we begin, we go over the classical
setting.

\subsubsection{Classical Setting} \label{define-gamma}
Let Charlie and Diana be two classical provers who are faced with
a classical verifier. The verifier sends Charlie a random clause
$c$, and Diana a random variable $v$ which appears in $c$. Charlie
should answer with the values that some satisfying assignment
gives the variables in $c$, and Diana should answer with the value
that same assignment gives $v$.  If the original formula is
$\gamma$-distant from being satisfiable, then the success
probability of Charlie and Diana is bounded by $1 -
\frac{\gamma}{3}$.

We prove a reduction from the quantum case to the classical one.
First, a simple lemma.

\begin{lem} \label{parseval}
If $\braket{u}{v} \leq 1/2$ and $|u|=|v|=|w| = 1$ then
$\braket{u}{w}
> 1-\epsilon \Rightarrow \braket{v}{w} < 1/2 + \frac{\sqrt{3 \epsilon}}{2} - \frac{\epsilon}{2}$.
\end{lem}
The proof follows from Taylor's approximation. A specific case: if
$\epsilon < 0.01$ then the final term is less than $0.99$.
Finally, let $\failprob(y, \tilde y, x, \tilde x, k)$ denote the
probability that the provers failed to convince the verifier,
given the measurement results $(y, \tilde y, x, \tilde x)$ and
$k$.

%

\begin{lem} \label{classical-lemma}
If there exists an index $k$, matrices $A_k$, $B_k$ and a set of
clauses $R \subset C$ such that
\begin{enumerate}
\item $|R| \ge (1 - \epsilon_{1}) M$

\item $\forall y \in R : \forall x \in y: |\{(\tilde y, \tilde x)
\in C \times V : \failprob(y, \tilde y, x, \tilde x,k) >
\epsilon_3 \}|
 < \epsilon_2 MN$

\item $\epsilon_3 < 1/200$
\end{enumerate}
then there is a classical strategy for Charlie and Diana which
gives them a success probability of at least $(1 - \epsilon_1)(1 -
\epsilon_2)(1 - \epsilon_3)(1 - 200 \epsilon_3)^2$.
\end{lem}

\begin{proof}
Charlie gets as an input a clause $y$ from the verifier. He
chooses a random $\tilde y$, and simulates Alice, conjugating by
$A_k$. Then he finds the closest possible legal classical
description to the state, by choosing $T(y), T(\tilde y)$ to
maximize the fidelity.
Similarly, Diane simulates Bob with her input $x$.

The classical verifier chooses independently, and therefore with
probability at least $1 - \epsilon_{1}$ he chooses a clause from
$R$. With probability greater than $1 - \epsilon_2$ the provers
choose a pair $\tilde y, \tilde x$ for which Alice and Bob have
good success probability. If this is the case, Alice and Bob's
success probability is at least $1 - \epsilon_3$. Since Alice
passes the SWAP test with probability $1 - \epsilon_3$, the state
she sends in the first step must pass the SWAP test with the
classical description she sent in the second step with probability
$1 - \epsilon_3$. If the latter is not the closest possible
classical description, then her probability of failing the SWAP
test, using $\delta = 0.99$ to ensure that there are no closer
alternatives, is at least $(1-\delta^2)/2 > (1-0.99^2)/2 > 1/200$.
Thus, the probability that this occurs is bounded by $200
\epsilon_3$. So with probability at least $1 - 200 \epsilon_3$
Alice and Charlie send the same assignment; given that, there is
only an $\epsilon_3$ chance of failure for Charlie. Finally, the
same argument as before applies to Diane (simulating Bob), which
contributes another factor of $1 - 200 \epsilon_3$ (the other
factors have already been counted for both provers).
\end{proof}

\begin{lem}
If the failure probability of Alice and Bob given result $k$ is
less than $\frac{\gamma^3}{5.55 \cdot 10^{13}}$ then there exists
a set $R$ with the properties stated in Lemma
\ref{classical-lemma}, with $\epsilon_1 = 0.003 \gamma$,
$\epsilon_2 = \gamma 10^{-3}$ and $\epsilon_3 = \gamma 10^{-4}$.
\end{lem}

\begin{proof}
Since the failure probability is less than $\frac{1}{6.96 \cdot
10^9}$, we must have, by Theorem \ref{large-NMW}, that $NM\tilde W
< 100 M W_A W_B$. By Lemmas \ref{small-nmw-clause-steps},
\ref{small-nmw-lem2} and \ref{bob-caught}, we have a set $H$ such
that $|H| \ge (1-0.002\gamma)M$, and
\[ \forall y \in H : \forall x
\in y: A(y) > W_A / (5M) \wedge B(x)> W_B / (5N) \]

\noindent Using $(\ref{prob-bound})$, this means that for any
tuple $(y, \tilde y, x, \tilde x) \in H$

$$\Pr(y, \tilde y, x\ \tilde x|k) \ge \frac{A(y)B(x)}{222MW_A W_B}
\ge \frac{W_A W_B}{25M N \cdot 222 M W_A W_B} = \frac{1}{5.55
\cdot 10^3 N M^2}$$

Denote $L(y,x) = \{(\tilde y, \tilde x): \failprob(y,\tilde y, x,
\tilde x,k) > 10^{-4} \gamma \}$, and $H_{\fail} = \{y \in H :
\exists x \in y: |L(y,x)| > 10^{-3} \gamma NM\}$. For any clause
$y \in H_{\fail}$, let $\fail(y) \in y$ denote the variable in $y$
for which $L(y,x)$ is maximal. We  bound Alice and Bob's failure
probability from below, to get an upper bound on $|H_{\fail}|$
\begin{eqnarray*}
\textrm{Pr}(\textrm{Provers fail to cheat}) & \ge & \sum_{y \in
H_{\fail}, x \in y} \sum_{(\tilde y, \tilde x) \in L(y,x)}
\failprob(y,\tilde y, x, \tilde x,k) \Pr (y, \tilde y, x, \tilde x
| k)
\\
& \ge & \sum_{y \in H_{\fail}} \sum_{(\tilde y, \tilde x) \in
L(y,x)} \gamma 10^{-4} \Pr (y, \tilde y, \fail(y),  \tilde x : k) \\
& \ge & \sum_{y \in H_{\fail}} \gamma 10^{-4} |L(y,\fail(y))| \Pr
(y, \tilde y, \fail(y),  \tilde x : k)
\\
& \ge & \sum_{y \in H_{\fail}} \frac{\gamma^2 NM W_{A} W_{B}}{25NM
\cdot 10^7\cdot 222M W_{A} W_{B}} = \frac{\gamma^2 |H_{\fail}|}{M
\cdot 5.55 \cdot 10^{10}}
\end{eqnarray*}
\noindent Where the last inequality comes from taking a tuple in
$H$. As Pr(The provers fail) $< \frac{\gamma^3}{5.55 \cdot
10^{13}}$, we have

\[ |H_{\fail}| <  \frac{\gamma^3}{5.55
\cdot 10^{13}} \cdot \frac{M \cdot 5.55 \cdot 10^{10}}{\gamma^2} =
10^{-3} \gamma M
\]
Taking $R = H \backslash H_{\fail}$, we get $|R| \ge (1 - 0.002
\gamma) M - |H_{\fail}| \ge (1 - 0.003 \gamma) M$ as required.
\end{proof}


\noindent {\it Proof of theorem \ref{main-thm}. } Assume $\Phi$ is
not satisfiable, and assume by contradiction that the provers had
some strategy which would work with success probability $\ge 1 -
\frac{\gamma^3}{5.55 \cdot 10^{13}}$. Then there has to be a
measurement result $k$ such that the success probability given $k$
is at least $1 - \frac{\gamma^3}{5.55 \cdot 10^{13}}$. However,
according to the previous lemma, either the provers are caught
with probability greater than $\frac{\gamma^3}{5.55 \cdot
10^{13}}$ (which contradicts our assumption on the success
probability), or there exists a set $R$ as in the premises of that
lemma. However, this would imply that there is a strategy in the
classical protocol with success probability $> (1 - 0.003 \gamma)
(1- \gamma 10^{-3})(1 - \gamma 10^{-4})(1 - 200 \gamma 10^{-4})^2
> 1 - \gamma/3$, which is a contradiction. \qedsymbol
\section{Conclusions and Open Problems}
We have shown that NEXP can be recognized in a quantum MIP
protocol, even if the provers have unlimited classical
communication between them. Our protocol achieves perfect
completeness and a constant gap. It only sends $O(\log(N))$
qubits, and thus can also be used for NP-complete languages with a
polylogarithmic communication. Some interesting questions still
remain open:

\begin{itemize}

\item What is the correct upper bound on the power of this proof
system? Note that if the provers were allowed to make any joint
separable measurement it would be exactly NEXP. Does adding
provers or communication rounds help? What happens if there is
just one quantum round?

\item Is there a parallel repetition lemma for protocols when the
provers are allowed to communicate with each other? The original
proof of \cite{Raz95} does not apply here.

\item What happens in the dual problem, when the provers are
allowed to share entanglement but are not allowed to communicate?
Does our protocol still work, with a different proof?

\item Does our result hold when the provers have a bounded amount
of entanglement in addition to their communication channel?

\end{itemize}

\section{Acknowledgments}\label{sec:ack}
A. H. wishes to thank Scott Aaronson, Dorit Aharonov, Richard
Cleve,  Irit Dinur, Julia Kempe, Debbie Lueng, Oded Regev, Peter
Shor, John Watrous and many others for their generous help.

\newpage
\pagebreak \setcounter{page}{1}

\appendix

\section{Calculating Probabilities}\label{calc-prob}
Let $r = (y, \tilde y, x, \tilde x)$. We wish to estimate
$\Pr(r|k)$.  Bayes' rule gives
\[\Pr(r | k) = \frac{\Pr(k|r) \Pr(r)}{\Pr(k)} = \frac{\Pr(k|r) \Pr(r)}{\sum_{s} \Pr(k |s) \Pr(s)}
\]

\noindent where $s$ denotes any legal tuple $s = (c, \tilde c, v,
\tilde v)$ with $c, \tilde c \in C$, $v, \tilde v \in V$ and $v
\in c$. As the prior distribution for all legal tuples is
identical, we are only interested in calculating $\Pr(k|s)$ for
any legal tuple $s = (c, \tilde c, v, \tilde v)$.

In the protocol we presented, the provers first apply their
measurement and get $k$, and then the verifier measures to get
$s$. However, it is physically equivalent to assume the verifier
measured first. As the states sent to the provers are unentangled
after tracing out the verifier, we have that

\[\Pr(k|s) = \Tr((I \otimes A_k) \rho_A (I \otimes A_k )^\dagger)
\cdot \Tr((I \otimes B_k) \rho_B (I \otimes B_k )^\dagger) \]

Where $\rho_A$ is the state in $H^v_A \otimes H^M_A$, $\rho_B$ is
the state in $H^v_B \otimes H^M_B$, and the identity is applied on
the verifier's side.

When considering states in $H^v_A \otimes H^m_A \otimes H^p_A$, we
stick to the convention that the first $m$ qubits define the
verifier's private space, then next $m+3$ describe the message
qubits, and the last $d$ define Alice's private space. We can now
calculate

\[A_k(y) = \Tr(A_k (\ket{y}
\bra{y} \otimes I)A_k) = \sum_{j = 1}^{8Md}
\sum_{h=8d(y-1)+1}^{8Dy} A_k[j,h] \overline{A_k[j,h]} = \sum_{j =
1}^{8Md} \sum_{h=8d(y-1)+1}^{8Di} |A_k[j,h]|^2\]

\[B_k(x) = \Tr(B_k (\ket{x}
\bra{x} \otimes I)B_k) = \sum_{j =
1}^{2Nd}\sum_{h=2d(x-1)+1}^{2dx} B_k[j,h] \overline{B_k[j,h]} =
\sum_{j = 1}^{2Nd}\sum_{h=2d(x-1)+1}^{2di} |B_k[j,h]|^2\]

\noindent We now assume that the $x, \tilde x$ is being traced out,
and only look at the probabilities for $y, \tilde y$,
generated from $\Tr((I \otimes A_k) \rho_A (I \otimes A_k
)^\dagger)$. As $A_k(y)$ is just the trace out of the private data
and the qubits which fit the assignment, then $A_k(y) = \Tr((I
\otimes A_k) \rho_A (I \otimes A_k )^\dagger)$. We are analyzing
the following expression:
\[ \Tr (I_{M} \otimes A_{8Md} \rho_A (I_{M} \otimes A_{8Md})^{\dagger}) \]
Up to normalization, $\rho_A$ is a matrix which contains exactly
four $1$s, arranged: $(a,a), (a,b), (b,a), (b,b)$. However, as we
shall soon see, either $a=b$ (in which case we have a single cell
with a $4$ in it) or else $|a-b| \geq 8Md$ and thus, by the
previous paragraph, we can ignore the off-diagonal entries. In
both cases we can restrict our attention to the diagonal entries.

Thus the structure of the $\rho_A$ matrix is:
\[ \rho = \frac{1}{\sqrt{2}} (\ket{yy} + \ket{\tilde y \tilde y})
\otimes \ket{000} \bra{000} (\bra{yy} + \bra{\tilde y \tilde y})
\frac{1}{\sqrt{2}} \otimes \ket{0_d} \bra{0_d} \in H^A_{v} \otimes
H^A_M \otimes H^A_p
\]
Note that the term $0_d$ refers to element in a space of dimension
$d$, as opposed to $000$, an element in a space of dimension
$2^3$. If $y=\tilde y$ then obviously there is only one nonzero
cell in the final matrix, on the diagonal. Otherwise, since
$\ket{yy}$ is located in the cell $My+y = (M+1)y$, and $\tilde y
\neq y$, they are differentiated (after tensoring) by at least
$(M+1) \cdot 8 \cdot d > 8Md$, as required.

Let  $A_k(i) = \sum_{j = 1}^{8Md} \sum_{h=8d(i-1)+1}^{8Di}
A_k[j,h] \overline{A_k[j,h]} = \sum_{j = 1}^{8Md}
\sum_{h=8d(i-1)+1}^{8Di} |A_k[j,h]|^2$. The probability that the
verifier measures $y, \tilde y$ in the modified protocol given $k$
is
\begin{eqnarray*}
P(y,\tilde y | k) & = & \frac{P(k | y, \tilde y) P(y, \tilde
y)}{P(k)} \\
& = & \frac{P(k | y, \tilde y) P(y, \tilde y)}{\sum_{z, \tilde z}
P(k | z, \tilde z) P(z, \tilde z)} \\
& = & \frac{\Tr (A_k \rho_{y, \tilde y} A_k^\dagger)}{\sum_{z,
\tilde z} \Tr (A_k \rho_{z, \tilde z} A_k^\dagger)} \\
\mbox{(equal unless $y=\tilde y$) \hspace{2cm}} & \geq &
\frac{A_k(y) + A_k(\tilde y)}{\sum_{z \neq \tilde z} (A_k(z) +
A_k(\tilde z)) + \sum_z 4 A_k(z)} \\
& = & \frac{A_k(y) + A_k(\tilde y)}{\sum_{z, \tilde z} (A_k(z) +
A_k(\tilde z)) + \sum_z 2 A_k(z)} \\
& = & \frac{A_k(y) + A_k(\tilde y)}{2 M W_{A_k} + 2 W_{A_k}}
\end{eqnarray*}
where $W_{A_k}$ is the total weight: $W_{A_k} = \sum_z A_k(z)$.
Note that if $y= \tilde y$ we use $4A_k(y)$ instead of $A_k(y) +
A_k(\tilde y)$.

\subsection{Bounding the Denominator}

Let $W_{A_k} = \sum_{i}A_k(i), W_{B_k} = \sum_{i}B_k(i), \tilde W
= \Sigma_{c \in C, v\in c} A_k(c)B_k(v)$. We want to bound the
denominator in

\begin{eqnarray*}
\Pr (y, \tilde y, x, \tilde x |k) = \frac{(A(y) + A(\tilde
y))(B(x) + B(\tilde x))}{\sum_{c, \tilde c \in C, v \in c, \tilde
v \in V} \Pr (c, \tilde c, v, \tilde v |k)}
\end{eqnarray*}


Note that if $c = \tilde c$, then $\Tr((I \otimes A_k) \rho_A (I
\otimes A_k )^\dagger) = 4 A_k(c)$. However, when $c \neq  \tilde
c$, we account this twice (because any of them can be considered
first in the sum). Thus, the denominator becomes

\begin{equation} \label{ctildec} \sum_{c, \tilde c} \sum_{v \in c, \tilde v} (A_k(c) +
A_k(\tilde c))  (B_k(v) + B_k(\tilde v)) + 2 \left( \sum_{c =
\tilde c, v, \tilde v} + \sum_{c, \tilde c, v = \tilde v}\right) +
4 \sum_{c = \tilde c,v = \tilde v}
\end{equation}

\noindent where all the sums are on $(A_k(c) + A_k(\tilde c))
(B_k(v) + B_k(\tilde v))$, and factors of two and four come from
$c = \tilde c$, and $v = \tilde v$. We begin by bounding the first
two sums (which will contribute most of the weight). We omit the
subindex $k$.

\[ \sum_{c, \tilde c} \sum_{v \in c, \tilde v} (A(c) + A(\tilde
c))  (B(v) + B(\tilde v)) = \sum_{c, \tilde c} \sum_{v \in c,
\tilde v} A(c)B(v) + A(c)B(\tilde v) + A(\tilde c)B(v) + A(\tilde
c)B(\tilde v) \]

We now look at each of the four terms separately:

\[ \sum_{c, \tilde c} \sum_{v \in c, \tilde v} A(c)B(v) = MN \sum_{v \in c, \tilde v} A(c)B(v) MN \tilde
W \]
\[
\sum_{c, \tilde c} \sum_{v \in c, \tilde v} A(c)B(\tilde v) = 3M
\sum_{c \in C, \tilde v \in V} A(c)B(\tilde v) = 3M W_A W_B \]

\[ \sum_{c, \tilde c} \sum_{v \in c, \tilde v} A(\tilde c)B(v) = 5N
W_A W_B < 5M W_A W_B
\]

\noindent And $\sum_{c, \tilde c} \sum_{v \in c, \tilde v}
A(\tilde c)B(\tilde v) = 3M W_A W_B$. We used the fact that $\Phi$
is $3-SAT$, and that each variable appears exactly 5 times.

We return to bounding the sums in $(\ref{ctildec})$. By fixing
$c$, we get that if $c = \tilde c$ the second sum is bounded,
relative to the first, by a factor of $2/M$. Fixing $\tilde v$, we
can bound the third sum by a factor of $2/N$. Fixing both, the
fourth sum is bounded by a factor of $4/(MN)$. We get an overall
bound for the denominator of:
\[
(M N \tilde W + 3 M W_A W_B + 5 N W_A W_B + 3 M W_A W_B) (1 + 2/M
+ 2/N + 4/(MN))
\]
%

Since $M$ and $N$ are arbitrarily large, and $M \geq N$, we deduce
our bound:
\[ 2 (M N \tilde W + 11 M W_A W_B) \]

which finally gives

\begin{equation*}
\Pr (y, \tilde y, x, \tilde x |k) \geq \frac{A(y)B(x) +
A(y)B(\tilde x) + A(\tilde y)B(x) + A(\tilde y)B(\tilde x)}{2 MN
\tilde W + 22 M W_{A} W_{B}}
\end{equation*}

\section{Proof of Lemma \ref{alice-damage}}\label{app-alice-damage}

{\bf Lemma \ref{alice-damage}} {\it Assume $A(y) \ge p A(\tilde
y)$, $p > 1$. Then for any assignment $T$, the probability that
the verifier will catch Alice cheating in the SWAP test is at
least $\frac{1}{2} - \frac{\sqrt{p}}{1+p}$.}

\begin{proof}
Let $\sigma = \Tr_{H^{p}_A} \frac{(I \otimes A_k) \rho (I \otimes
A_k)^\dagger}{\Tr((I \otimes A_k) \rho (I \otimes A_k)^\dagger)}$,
and $\ket{\psi} = 1/\sqrt{2}(\ket{yy} \ket{T(y)} + \ket{\tilde y
\tilde y} \ket{T(\tilde y)})$. Taking $\delta = \sqrt{\bra{\psi}
\sigma \ket{\psi}}$ the fidelity between $\ket{\psi}$ and $\rho$,
the SWAP test has probability at least $\frac{1 - \delta^2}{2}$ to
distinguish between them \cite{BCWW01}.

To calculate $\sigma$, we utilize the result in Appendix
\ref{calc-prob}. Since $\rho$ consists of four elements in a
rectangle $((8M+8)y,(8M+8)y), ((8M+8)y,(8M+8)\tilde y),
((8M+8)\tilde y,(8M+8)y), ((8M+8)\tilde y,(8M+8)\tilde y)$,
differentiated by a distance of at least $8M$, the nondiagonal
elements do not contribute to the trace.


For a given assignment $T(y) \in \{0,1\}^3$ and $T(\tilde y) \in
\{0,1\}^3$, let $\ket{\psi} = 1/\sqrt{2}(\ket{yy} \ket{T(y)} +
\ket{\tilde y \tilde y} \ket{T(\tilde y)})$. The fidelity between
the pure state $\psi$ and the quantum state is $\sqrt{\bra{\psi}
\sigma \ket{\psi}}$.

$\ket{\psi}$ is an equal superposition of two base vectors, one
corresponding to the base state $\ket{yyT(y)}$ and the other to
$\ket{\tilde y \tilde y T(\tilde y)}$. Thus the multiplication is
effectively the sum of four elements arranged in a rectangle
(multiplied by $1/2$). To calculate each of these four elements,
we turn to Appendix \ref{calc-prob}. Since, in the tensor product
$I \otimes A_k$, any cell whose two coordinates differ by at least
$8M$ is zero, we can simplify and get:

\begin{eqnarray*}
\sigma[yy T (y), yy T (y)] = \Tr(A \ket{y T(y)} \bra{y T(y)}A)
\\
\sigma[yy T (y), \tilde y \tilde y T (\tilde y)] = \Tr(A \ket{y
T(y)} \bra{\tilde y T(\tilde y)}A)
\end{eqnarray*}

We write the elements, sticking to the convention that the first
$m$ qubits describe the verifier's private space, the next $m$ fit
the clause in the message space and the last three fit the value
of the assignment:

\begin{eqnarray*}
\sigma(8My + 8y + a, 8My + 8y + a) & = & \sum_{i=1}^d \sum_{j=1}^d |A[8dy + da + i,8dy + da + j]|^2 \\
\sigma((8M+8)\tilde y + b, (8M+8)\tilde y + b) & = & \sum_i \sum_j |A[8d\tilde y + db + i,8d \tilde y + db + j]|^2 \\
\sigma((8M+8)y + a, (8M+8)\tilde y + b) & = & \sum_i \sum_j A[8dy
+ da + i,8dy + da +j]
\overline{A[8d\tilde y + db + i, 8d \tilde y + db + j]} \\
\sigma((8M+8)\tilde y + b, (8M+8)y + a) & = & \sum_i \sum_j A[8d
\tilde y + db + i, 8d \tilde y + db + j] \overline{A[8dy + da + i,
8dy + da + j]}
\end{eqnarray*}
Note that as $A A^\dagger$ is a measurement operator, we have that
$A(y) \leq 1$, so $A(\tilde y) \leq 1/p$. Now calculating,
reindexing by $s=8My + 8y + T(y)$ and $t = 8M\tilde y + 8 \tilde y
+ T(\tilde y)$, and folding the sum into the expression, we get:
\begin{eqnarray*}
\frac{|\sigma[s,s]|^2 + |\sigma[t,t]|^2 + \sigma[s,t]
\overline{\sigma[t,s]} + \sigma[t,s]
\overline{\sigma[s,t]}}{2(A(y) + A(\tilde y))} & \leq & \frac{A(y)
+ A(\tilde y) + 2 \sqrt{A(y) A(\tilde y)}}{2(A(y) + A(\tilde y))} \\
& = & \frac{1}{2} + \frac{\sqrt{A(y) A(\tilde y)}}{A(y) + A(\tilde
y)} < 1
\end{eqnarray*}
The last inequality follows from the AM-GM inequality. More
precisely, since the ratio is at least $p$, the extreme value is
achieved when it is exactly $p$, which (when substituting) gives
what we need. When $p \ge \sqrt{2}$, this gives $1/2 + 2^{1/4} /
(1 + \sqrt{2}) \le 0.993$, as required. When $p \ge 2$, we get
$1/2 - \sqrt{2}/3 \le 0.975$.
\end{proof}

\section{Proofs for Lemmas \ref{steps} and \ref{T-steps}} \label{app-large-nmw}
Remember that for $c \in C$,  $u(c) = \Sigma_{v \in c} A(c)B(v)$,
and for $S \subset C$, $U(S) = \sum_{c \in S} u(c)$. We also
defined
\[
S_{i} = \left\{c : \frac{\tilde W}{2^{i+1}} < u(c) \le
\frac{\tilde W}{2^{i}} \right\}
\]

\noindent {\bf Lemma \ref{steps}.} {\it If there exists an index
$j$ such that $\sum_{i=0}^{j-1} U(S_i)
>\tilde W / 100$ and $\sum_{i=j+1}^{\infty} U(S_i) > \tilde W /
100$, then the provers get caught with constant probability
$\frac{1}{6.96 \cdot 10^9}$, generated from a $(\frac{1}{4.8 \cdot
10^{7}}, \sqrt{2})$ bad set.}

\begin{proof}
We construct such a bad set $D$. For any clause $c$ in variables
$v_1, v_2, v_3$, let $\vmax(c)$ denote the variable $v_i \in c$
such that $B[v_i] = \max\{B[v_j] : j = 1,2,3\}$, and define
$\vmin(c)$ analogously. Let $S_{\ups} = \cup_{i=0}^{j-1} S_i$, and
$S_{\downs} = \cup_{i=j+1}^{\infty} S_i$. As $U(S_{\downs}) >
\tilde W /100$,  and $S_{down}$ consists of ``light'' clauses, we
must have $|S_{\downs}|
> M/100$. Partition $S_{\downs}$ arbitrarily into two sets $S_{l}$
and $S_{r}$, such that $|S_{l}|, |S_{r}| \ge M/200$. The idea is
that each clause in $S_{\ups}$ will contribute $|S_{l}| \cdot
|S_{r}|$ elements to $D$.

For each $c_{\ups} \in S_{\ups}, c_{l} \in S_{l}, c_{r} \in
S_{r}$, we have $u(c_{\ups}) > 2u(c_l), u(c_{\ups}) > 2u(c_r)$.
Taking the maximal element in the sum for $c_{\ups}$, and the
minimal element for $c_l, c_r$, we get:
\begin{eqnarray*}
A(c_{\ups}) B(\vmax(c_{\ups})) & > & 2A(c_{l}) B(\vmin(c_{l})) \\
A(c_{\ups}) B(\vmax(c_{\ups}))  & > & 2A(c_{r}) B(\vmin(c_{r}))
\end{eqnarray*}
Assume WLOG that $A(c_l) < A(c_r)$. Then
\begin{eqnarray*}
A(c_{\ups}) B(\vmax(c_{\ups})) & > & 2A(c_{l}) B(\vmin(c_{r}))
\end{eqnarray*}
So in the tuple $(c_{\ups}, c_{l}, \vmax(c_{\ups}),
\vmin({c_{r}}))$ at least one of the provers damages the state by
at least $\sqrt{2}$. We add this tuple to $D$. Note that we have
added a (distinct) element to $D$  for each of the $|S_l| \cdot
|S_r|$ choices of $c_l, c_r$, as desired. Let $D_\ups$ denote the
elements contributed to $D$ by $c_{\ups}$.

The next step is to prove that $D$ has constant probability. Note
that under the conditions of the lemma, we have
\[
\sum_{y, \tilde y \in C, x \in y, \tilde x \in V} \Pr(y,\tilde y,
x, \tilde x|k) < 22 M W_{A}W_{B} + 2 NM \tilde W < 4NM \tilde W
\]
\begin{eqnarray*}
\sum_{(y, \tilde y, x, \tilde x) \in D} \Pr((y, \tilde y, x,
\tilde x) : k) & = & \sum_{c_{\ups}} \sum_{(\tilde y, \tilde x)
\in D(c_{ups})} \Pr ((c_{\ups}, \tilde y , \max(c_{\ups}), \tilde
x) : k) \\
& \ge & \frac{1}{4NM \tilde W} \sum_{c_{\ups}} \sum_{(\tilde y,
\tilde x) \in D(c_{\ups})} A(c_{\ups})B(\max(c_{\ups}))
\\
& \geq & \frac{1}{4NM \tilde W}  |S_{l}| \cdot |S_{r}|
\sum_{c_{\ups}} A(c_{\ups})B(\max(c_{\ups}))
\\
& \ge & \frac{1}{4NM \tilde W}  |S_{l}| \cdot |S_{r}|
\sum_{c_{\ups}} u(c_{\ups})/3
\\
& \ge & \frac{1}{4NM \tilde W}  |S_{l}| \cdot |S_{r}|
\tilde W / 300 \\
& \ge & \frac{1}{4NM \tilde W} \cdot \frac{M}{200}
\cdot \frac{M}{200} \cdot \frac{\tilde W}{300} \\
\mbox{(because $M>N$)\hspace{2cm}} & > & \frac{NM \tilde W}{4NM
\tilde W \cdot 200 \cdot
200 \cdot 300} \\
& = & \frac{1}{4.8 \cdot 10^{7}}
\end{eqnarray*}
\end{proof}

Remember that $F = S_j \cup S_{j+1}$ and $T_{i} = \{c \in F:
\frac{W_{A}}{2^{i+1}} < A(c) \le \frac{W_A}{2^{i}} \}$. We wish to
prove

\noindent {\bf Lemma \ref{T-steps}.} {\it If there exists an index
$j$ such that $\sum_{i=0}^{j-1} |T_i|
> |F|/ 100$, and $\sum_{i=j+1}^{\infty} |T_i| > |F| / 100$, then
the first prover gets caught with constant probability
$\frac{1}{4.2 \cdot 10^7}$, generated from a  $(\frac{1}{1.2 \cdot
10^{6}},2)$ bad set.}

\begin{proof}
Let $T_{\ups} = \cup_{i=0}^{j-1}T_{i}$, $T_{\downs} = \cup_{i=j +
1}^{\infty} T_{i}$. Note that any clause from $T_{\ups}$ at least
2-damages any clause in $T_{\downs}$. Take $D = \cup_{c_{\ups} \in
T_{\ups}} \{ \{c_{\ups} \} \times T_{\downs} \times \{
\vmax(c_{\ups}) \} \times V \}$. Note that $|T_{\downs}| > 0.98M /
100 > M /200$, and as $T_{\ups} \subset F$, we have $U(T_{\ups})
\ge \frac{0.98 \tilde W}{400} \ge \frac{\tilde W}{500}$, and thus
\[\Pr(D) \geq \frac{1}{4NM \tilde W}
|T_{\downs}| N \frac{\tilde W }{1500} \ge \frac{1}{1.2 \cdot
10^{6}}
\]
\end{proof}

\section{Proofs for Lemmas \ref{small-nmw-clause-steps}, \ref{small-nmw-lem2} and \ref{bob-caught}}\label{app-small-nmw}
The proofs in this appendix are very similar and very easy. We
recall some definitions, then state the lemmas. Define $S_i = \{ c
\in C : \frac{W_A}{2^{i+1}} \le A(c) < \frac{W_{A}}{2^i} \}$. For
a set $S \subset C$, let $W(S) = \Sigma_{c \in S} A(c)$.

\noindent {\bf Lemma \ref{small-nmw-clause-steps}.} {\it If
$NM\tilde W < 100 M W_A W_B$ and there exists an index $i$ such
that
\[\sum_{j = 0}^{i-1} W(S_j) > \gamma 10^{-4} W_A \bigwedge
\sum_{j = i + 1}^{\infty} |S_j| > \gamma 10^{-4} M\]
\noindent
then Alice is caught cheating with probability
$\frac{\gamma^2}{2.6 \cdot 10^{12}}$, generated from a
$(\frac{\gamma^2}{7.4 \cdot 10^{10}}, 2)$ bad set.}

\begin{proof}
Let $S_{\ups} = \cup_{j=0}^{i-1} S_j$, $S_{\downs} =
\cup_{j=i+1}^{\infty} S_j$. Let $D = \cup_{c \in S_{\ups}} \cup_{v
\in c} \{c\} \times S_{\downs} \times \{v\} \times V$. Every $(y,
\tilde y, x, \tilde x) \in D$ is $2$-damaged by Alice. On the
other hand,
\[
\Pr (y , \tilde y, x  , \tilde x | k)
\stackrel{(\ref{prob-bound})}{\ge} \frac{(A(y)+A(\tilde y))(B(x) +
B(\tilde x))}{22 M W_{A} W_{B} + 2 NM \tilde W} \ge
\frac{A(y)B(\tilde x)}{22 M W_{A} W_{B} + 2 NM \tilde W} \ge
\frac{A(y)B(\tilde x)}{222 M W_{A} W_{B}}
\]
Summing this over $D$ gives
\[
\Pr(D) \ge \sum_{y \in S_{\ups}} \sum_{x \in y} \sum_{\tilde y \in
S_\downs} \sum_{\tilde x \in V} \frac{A(y)B(\tilde x)}{222 M W_{A}
W_{B}} \geq \sum_{y \in S_{\ups}} \frac{3 \cdot 10^{-4} \gamma M
W_{B} A(y)}{222M W_{A} W_{B}} \ge \frac{3 \gamma^2 W_{A}}{222
\cdot 10^8 W_{A}} \ge \frac{\gamma^2}{7.4 \cdot 10^{10}}
\]
\end{proof}

\noindent{\bf Lemma \ref{small-nmw-clause-steps}.} {\it If
$NM\tilde W < 100 M W_A W_B$ and the second condition of Lemma
\ref{small-nmw-clause-steps} does not hold, then there exists an
index $i$ such that for $F = S_{i} \cup S_{i+1}$ we have $|F| \ge
(1 - 0.0002 \gamma) M \bigwedge W(F) \ge (1 - 0.0002 \gamma) W_A
\bigwedge \forall c \in F : A(c) \ge \frac{W_{A}}{5M} $.}

\begin{proof}
Choose $t$ to be the smallest index for which the first half of
the condition does hold, i.e., $\sum_{j=0}^{t-1} W(S_j) > \gamma
10^{-4} W_A$. Then the second half of the condition cannot hold,
i.e.
\[ \sum_{j=t+1}^{\infty} |S_j| \le \gamma 10^{-4} M \]
Take $i=t-1$ (note that $t \neq 0$ because otherwise the first
half of the condition does not hold). So:
\[ |S_i| + |S_{i+1}| = M - \sum_{j=0}^{i-1} |S_j| -
\sum_{j=i+2}^{\infty} |S_j| \geq M - \sum_{j=0}^{i-1} |S_j| -
\gamma 10^{-4} M \geq M - \gamma 10^{-4} M - \gamma 10^{-4} M
\]
where the last inequality follows since the total weight
$\sum_{j=0}^{i-1} W(S_j) < \gamma 10^{-4} W_A$, but each clause in
the $S_j$s contributes at least $2^{-i} W_A$ to $W(S_j)$ while
each clause outside of the $S_j$s contributes at most $2^{-i-1}
W_A$. A similar argument now applies to the weight $W(S_i) +
W(S_{i+1})$. Finally, for each $c \in F$ we have
\[ A(c) \ge
\frac{W(F)}{4 |F|} \ge \frac{(1 - 0.0002 \gamma) W_{A}}{4M} \ge
\frac{W_{A}}{5M} \]
\end{proof}

Remember $S_i$: $T_i = \left\{ v \in V : \frac{W_B}{2^{i+1}} \le
B(v) < \frac{W_{B}}{2^i} \right\}$. We now prove

\noindent {\bf Lemma \ref{bob-caught}.} {\it Either Bob gets
caught cheating with probability $\frac{\gamma^2}{3.9 \cdot
10^{12}}$ which is generated from a $\frac{\gamma^2}{1.1 \cdot
10^{11}},2$ bad set, or else there exists an index $i$ such that
for $G = T_{i} \cup T_{i + 1}$ we have $|G| > (1 - 0.0002 \gamma)
N$, $\Sigma_{v \in G} B(v) \ge (1 - 0.0002 \gamma) W_{B}$ and for
each $v \in G$, $B(v) \ge \frac{W_{B}}{5N}$.}

\begin{proof}
If no such index exists then there is a separating index $i$
such that letting $T_{\ups} = \cup_{j=0}^{i-1} S_j$, $T_{\downs} =
\cup_{j=i+1}^{\infty} S_j$, we have $\sum_{v \in T_{\ups}} B(v) >
10^{-4} \gamma W_{B}$, $|T_{\downs}| > 10^{-4} \gamma N$. Let $D =
\cup_{v \in T_{\ups}} \cup_{c:v\in c} \{c\} \times C \times \{v\}
\times T_\downs$.

\[
\Pr(D:k) \ge \sum_{(y,x,\tilde y, \tilde x)} \frac{A(\tilde y)
B(x)}{222 M W_{A} W_{B}} \ge \sum_{x \in T_{\ups}} \frac{\gamma N
W_{A} B(x)}{222 M W_{A} W_{B}}
 \ge \frac{\gamma^2 N W_{B}}{2.22 \cdot 10^{10} M W_{A}}
\ge \frac{\gamma^2 M}{1.1 \cdot 10^{11} M} = \frac{\gamma^2}{1.1
\cdot 10^{11}}
\]
where we used the fact that each variable appears in the formula
$5$ times.
\end{proof}
%
%
%
%
%
%

%
%
%
\end{document}